\documentclass[reprint,pra,amssymb,amsmath,superscriptaddress,
longbibliography]{revtex4-1}

\pdfoutput=1

\usepackage{graphicx}
\usepackage[colorlinks, citecolor=blue, linkcolor=blue,
urlcolor=blue]{hyperref}

\newcommand{\ket}[1]{|{#1} \rangle}
\newcommand{\bra}[1]{{\langle {#1}|}}
\newcommand{\braket}[2]{\langle {#1} | {#2} \rangle}
\renewcommand{\phi}{\varphi}

\usepackage{amsthm}
\newtheorem{theorem}{Theorem}
\newtheorem{lemma}[theorem]{Lemma}
\newtheorem{corollary}[theorem]{Corollary}
\newtheorem{proposition}[theorem]{Proposition}

\begin{document}
\title{Geometry of entanglement in the Bloch sphere}
\author{Michel Boyer}
\email{boyer@iro.umontreal.ca}
\affiliation{D\'epartement IRO, Universit\'e de Montr\'eal,
Montr\'eal (Qu\'ebec) H3C 3J7, Canada}
\author{Rotem Liss}
\email{rotemliss@cs.technion.ac.il}
\author{Tal Mor}
\email{talmo@cs.technion.ac.il}
\affiliation{Computer Science Department, Technion, Haifa 3200003,
Israel}

\begin{abstract}
Entanglement is an important concept in quantum information,
quantum communication, and quantum computing.
We provide a geometrical analysis of entanglement
and separability for \emph{all} the rank-2 quantum mixed states:
complete analysis for the bipartite states,
and partial analysis for the multipartite states.
For each rank-2 mixed state, we define its unique Bloch sphere,
that is spanned by the eigenstates of its density matrix.
We characterize those Bloch spheres into exactly five classes
of entanglement and separability, give examples for each class,
and prove that those are the only classes.
\end{abstract}


\maketitle

\section{\label{intro} Introduction}
Entanglement is a very important property of quantum states,
relevant to the foundations of quantum mechanics
(e.g., the Einstein-Podolsky-Rosen paradox and Bell's inequality),
as well as to quantum information, quantum communication
(including quantum teleportation and quantum cryptography),
quantum computers and simulators, and quantum many-body systems.

The relations between entanglement, partial transpose,
and non-classical correlations between the subsystems,
are well-understood for pure quantum bipartite states.
However, for mixed quantum states there are still many open questions.
Even bipartite mixed states of rank $2$
(namely, states that can be written as
$\rho = p \ket{\phi} \bra{\phi} + (1 - p) \ket{\psi} \bra{\psi}$,
where $0 < p < 1$, and $\ket{\phi}, \ket{\psi}$ are
bipartite orthonormal states and are the eigenstates of $\rho$),
that are discussed in this paper, are not well-understood.
Studying such states is thus a major challenge
in the field of mixed-state quantum entanglement.

It is known that if a mixed state does not have a positive
partial transpose then it is entangled
and presents a nonlocal behavior~\cite{Peres1996}.
However, one can find separable states presenting a nonlocal behavior
(e.g.,~\cite{BDFMRSSW99}),
and one can find entangled states that have a positive partial
transpose~\cite{Horodecki1996}; those states are bound entangled,
namely, their entanglement cannot be distilled~\cite{Horodecki1998}.
It was later proved that bound entangled states cannot
have rank $3$ or less~\cite{Horodecki2003,Chen2008}.
Therefore, checking whether a \emph{specific} rank-2 state
is entangled is trivial: it is entangled if and only if
it does not have a positive partial transpose;
however, in this paper we discuss the problem of
classifying each rank-2 state by checking which states
in its \emph{Bloch-sphere neighborhood} (namely, in
its corresponding Bloch sphere) are entangled.

Entanglement distillation (for pure states)~\cite{BBPS96}
and entanglement purification (for mixed states)~\cite{BBPSSW96}
are processes of distilling Bell states (or other maximally entangled
states) from some copies of an initial state.
An efficient protocol is known for pure states,
but not for mixed states.
This provides another motivation for studying and finding ways
to fully characterize the simplest non-pure bipartite states
(the rank-2 bipartite mixed states).

The notion of the Bloch sphere, also known as the Poincar\'e sphere,
is a very useful geometrical interpretation of a single qubit.
It can be extended to any $2$-dimensional (complex) \emph{subspace}
of a full Hilbert space -- for example, the subspace
spanned by the eigenstates of any given rank-2 mixed state.

We define here the ``Bloch-sphere entanglement'' of a quantum rank-2
bipartite state. This (informally) means that we define the sets of
separable states and of entangled states inside the \underline{unique}
Bloch sphere associated with this quantum state.
We provide some examples,
and we prove that the five classes we present exhaust all
the possibilities of ``Bloch-sphere entanglement''.
We briefly discuss going beyond bipartite states,
and we briefly present an interesting exception (from
the above classification) for the case of just two qubits.

We primarily use the Peres-Horodecki
criterion~\cite{Peres1996,Horodecki1996}:
\emph{if} for a state $\rho$ of the system $AB$, the operator
$\rho^{T_B}$ is not positive semidefinite (where $\rho^{T_B}$ is the
partial transpose of $\rho$ with respect to the the subsystem $B$),
\emph{then} $\rho$ is entangled.

It was shown in~\cite{Horodecki1996} that for systems of dimensions
$2 \otimes 2$, $2 \otimes 3$, or $3 \otimes 2$,
$\rho$ is entangled \emph{if and only if} $\rho^{T_B}$ is not positive
semidefinite. Yet in higher dimensions there are entangled states 
(that are bound entangled states) that have a positive
partial transpose~\cite{Horodecki1996,BDMSST1999}.

In Sec.~\ref{weaker_crit} we present a weaker entanglement criterion
that we will use for proving our claims, and in Sec.~\ref{prop_bloch}
we introduce several important properties of Bloch spheres to be used
in our proofs. In Sec.~\ref{classify_ent} we present a classification
of all rank-2 states into five classes, and in Sec.~\ref{proof_central}
we prove that no other classes exist. In Sec.~\ref{two_qubit_proof}
we prove that one of the classes does not exist in a specific case
(the two-qubit case). In Sec.~\ref{multipartite_ent} we generalize
some of our results to multipartite entanglement.
In Sec.~\ref{prev_works} we describe previous works in this area,
and in Sec.~\ref{conclusion} we conclude.

\section{\label{weaker_crit} A weaker entanglement criterion}
We will use this weaker entanglement criterion to prove our claims:
\begin{lemma}\label{PeresWeaker}
Let $\rho^{AB}$ be a state of a bipartite system.
\textbf{If} there are states
$\ket{\phi_A}$, $\ket{\phi_B}$, $\ket{\psi_A}$, and $\ket{\psi_B}$
such that $\bra{\phi_A\phi_B}\rho^{AB}\ket{\phi_A\phi_B} = 0$
and $\bra{\phi_A\psi_B}\rho^{AB}\ket{\psi_A\phi_B} \neq 0$,
\textbf{then} $\rho^{AB}$ is entangled. 
\end{lemma}
\begin{proof}
Let $\rho = \rho^{AB}$, $\ket{\phi} = \ket{\phi_A\phi_B^*}$,
and $\ket{\psi} = \ket{\psi_A\psi_B^*}$,
where $\ket{\phi_B^*}$ and $\ket{\psi_B^*}$ are obtained from
$\ket{\phi_B}$ and $\ket{\psi_B}$ by
replacing their amplitudes in the standard (computational) basis
by their complex conjugates.

We first need a property of $\rho^{T_B}$.
By definition, the partial transpose of
$C_{ijkl} = \ket{i}\bra{j} \otimes \ket{k}\bra{l}$ is
$C^{T_B}_{ijkl} = \ket{i}\bra{j}\otimes \ket{l}\bra{k}$, and the
partial transpose $\rho^{T_B}$ of $\rho$
is obtained by a linear extension.
Therefore, for $C_{ijkl}$ it holds that
\begin{eqnarray*}
\bra{\phi_A\phi^*_B} C^{T_B}_{ijkl} \ket{\psi_A\psi^*_B}
&=& \braket{\phi_A}{i} \braket{j}{\psi_A}
\braket{\phi^*_B}{l} \braket{k}{\psi^*_B} \\
&=& \braket{\phi_A}{i} \braket{j}{\psi_A}
\braket{\psi_B}{k} \braket{l}{\phi_B} \\
&=& \bra{\phi_A\psi_B} C_{ijkl} \ket{\psi_A\phi_B},
\end{eqnarray*}
and by linearity,
\begin{equation*}
\bra{\phi_A\phi^*_B}\, \rho^{T_B}\, \ket{\psi_A\psi^*_B}
= \bra{\phi_A\psi_B}\, \rho\, \ket{\psi_A\phi_B}.
\end{equation*}

If the condition of the Lemma is satisfied, then
$\bra{\phi_A\phi^*_B}\rho^{T_B}\ket{\phi_A\phi^*_B}
= \bra{\phi_A\phi_B}\rho\ket{\phi_A\phi_B} = 0$ and
$\bra{\phi_A\phi^*_B}\rho^{T_B}\ket{\psi_A\psi^*_B}
= \bra{\phi_A\psi_B}\rho\ket{\psi_A\phi_B} \neq 0$.
From Lemma~\ref{zelemma} it follows that
$\rho^{T_B}$ is not positive semidefinite, and thus
that $\rho$ is entangled.
\end{proof}
We declare this Lemma to be a ``weaker'' criterion because
it proves entanglement only for a subclass of all the states
satisfying the Peres-Horodecki criterion.

\begin{lemma}\label{zelemma}
If a Hermitian operator $A$ is positive semidefinite and
$\bra{\phi}A\ket{\phi} = 0$, then
$\bra{\phi}A\ket{\psi} = 0$ for all $\ket{\psi}$.
\end{lemma}
\begin{proof}
Let $A = \sum_i \lambda_i \ket{i}\bra{i}$ with $\lambda_i\geq 0$;
$\bra{\phi}A\ket{\phi} = \sum_i \lambda_i |\braket{\phi}{i}|^2 = 0$
and thus $\braket{\phi}{i} = 0$ if $\lambda_i \neq 0$.
It follows that $\bra{\phi}A\ket{\psi}
= \sum_i \lambda_i\braket{\phi}{i}\braket{i}{\psi} = 0$
for all $\ket{\psi}$.
\end{proof}

Lemma~\ref{zelemma} was presented by us (MB and TM)
in a conference~\cite{BM2014,*BBM2017}.

\section{\label{prop_bloch} Properties of subspaces and Bloch spheres}
In the next sections, we also use the following results,
that were mentioned in~\cite{Osterloh2008}:
\begin{lemma}\label{lemma_decomp}
Let $\mathcal{H}'$ be a subspace of a Hilbert space $\mathcal{H}$.
Let $\rho \in \mathcal{L}(\mathcal{H}')$
(i.e., $\rho$ can be decomposed as a mixture of
pure states from $\mathcal{H}'$).
If $\rho = \sum_j q_j \ket{\phi_j}\bra{\phi_j}$
is a decomposition of $\rho$ with
$\ket{\phi_j} \in \mathcal{H}$ and $q_j > 0$, then
$\ket{\phi_j} \in \mathcal{H}'$ for all $j$.
\end{lemma}
\begin{proof}
Let $\big\{\ket{\psi_i}\big\}_{i\in I'}$ be an orthonormal basis of
$\mathcal{H'}$, and let us extend it to an orthonormal basis
$\big\{\ket{\psi_i}\big\}_{i\in I}$ of $\mathcal{H}$ ($I' \subseteq I$).
Let $\ket{\phi_j} = \sum_{i\in I} a_{ji} \ket{\psi_i}$, with
$a_{ji} = \braket{\psi_i}{\phi_j}$. Then for all $i \in I \setminus I'$,
\begin{equation*}
0 = \bra{\psi_i}\rho\ket{\psi_i} =
\sum_j q_j \braket{\psi_i}{\phi_j}\braket{\phi_j}{\psi_i} =
\sum_j q_j|a_{ji}|^2 ,
\end{equation*}
implying that $a_{ji} = 0$ for all $i \in I \setminus I'$, and thus
$\ket{\phi_j} = \sum_{i\in I'} a_{ji} \ket{\psi_i} \in \mathcal{H}'$.
\end{proof}

\begin{corollary}\label{corollary_decomp}
If a rank-2 mixed state $\rho$ is inside a specific Bloch sphere,
then all the pure states in all of its decompositions lie on
the same Bloch sphere.
\end{corollary}

By using Corollary~\ref{corollary_decomp}, we get:

\begin{corollary}\label{corollary_unique}
If $\rho$ is a rank-2 mixed state,
then it lies inside a \underline{unique} Bloch sphere
(the uniqueness is up to a possible rotation of the sphere).
\end{corollary}

\begin{corollary}\label{corollary_separable}
If a rank-2 mixed state $\rho$ is separable, then there exist
at least two different pure separable states on its unique Bloch sphere.
\end{corollary}

\section{\label{classify_ent} Classification of Bloch-sphere entanglement}
In the rest of this paper we use
Lemma~\ref{PeresWeaker} (a ``weaker entanglement criterion''),
Lemma~\ref{zelemma} (a ``positive semidefinite operators condition''),
Corollary~\ref{corollary_unique} (the ``unique-Bloch-sphere
corollary''), and Corollary~\ref{corollary_separable}
(a ``separable states condition'')
in order to provide a classification of Bloch-sphere entanglement.
This is based on the following understanding:
if $\rho$ is a bipartite rank-2 mixed state that is a mixture of
pure states in the Hilbert space $\mathcal{H}_A \otimes \mathcal{H}_B$,
then according to Corollary~\ref{corollary_unique}, it lies inside
a unique Bloch sphere (the uniqueness is up to a possible rotation);
and this Bloch sphere corresponds to a $2$-dimensional subspace
of $\mathcal{H}_A \otimes \mathcal{H}_B$.

We present five different classes of $2$-dimensional subspaces
of a bipartite system, that are distinguished
by their Bloch-sphere entanglement:
(It is sufficient to consider only examples for which $\mathcal{H}_A$
is $2$-dimensional ($\mathcal{H}_2$) and $\mathcal{H}_B$ is either
$2$-dimensional ($\mathcal{H}_2$) or $3$-dimensional ($\mathcal{H}_3$).)
\begin{enumerate}
\item \label{class_sep} No entanglement at all

Example in $\mathcal{H}_2 \otimes \mathcal{H}_2$:
$\text{Span}\{\ket{00}, \ket{01}\}$ (Fig.~\ref{fig_sep})

\item \label{class_line_ortho} Entanglement everywhere on and inside
the sphere except a \underline{line} (of separable states) connecting
two \emph{orthogonal} pure states on the sphere (e.g., the poles)

Example in $\mathcal{H}_2 \otimes \mathcal{H}_2$:
$\text{Span}\{\ket{00}, \ket{11}\}$ (Fig.~\ref{fig_line_ortho})

\item \label{class_line_nonortho} Entanglement everywhere on and inside
the sphere except a \underline{line} (of separable states) connecting
two \emph{non-orthogonal} pure states on the sphere

Example in $\mathcal{H}_2 \otimes \mathcal{H}_2$:
$\text{Span}\{\ket{00}, \ket{++}\}$ (Fig.~\ref{fig_line_nonortho})

\item \label{class_point} Entanglement everywhere on and inside
the sphere except a \underline{single separable point} on the sphere

Example in $\mathcal{H}_2 \otimes \mathcal{H}_2$:
$\text{Span}\{\ket{00}, \alpha \ket{01} + \beta \ket{10}\}$
with $\alpha \beta \ne 0$ (Fig.~\ref{fig_point} and
Proposition~\ref{prop_class4_example})

\item \label{class_entang} Entanglement everywhere
(``completely entangled subspace'')

Example in $\mathcal{H}_2 \otimes \mathcal{H}_3$:
$\text{Span}\{[\ket{00}+\ket{11}]/\sqrt{2},
[\ket{02}+\ket{10}]/\sqrt{2}\}$ (Fig.~\ref{fig_entang} and
Proposition~\ref{prop_class5_example})

\underline{Does not exist} in $\mathcal{H}_2 \otimes \mathcal{H}_2$.
(Proof is given in Sec.~\ref{two_qubit_proof},
as Proposition~\ref{prop_two_qubit}.)
\end{enumerate}

Very similar examples can be found in all the bipartite Hilbert spaces
(if the dimensions of both subsystems are at least $2$),
except the example to Class~\ref{class_entang}, that does not exist in
$\mathcal{H}_2 \otimes \mathcal{H}_2$.

The analysis of Classes~\ref{class_sep}-\ref{class_line_nonortho}
(see Figures~\ref{fig_sep}-\ref{fig_line_nonortho})
is very simple and follows directly
from the proof of the general Theorem~\ref{central_theorem}.
Generally speaking, if two pure separable states exist
on the Bloch sphere, then it belongs to one of those classes.

\begin{figure}
\includegraphics{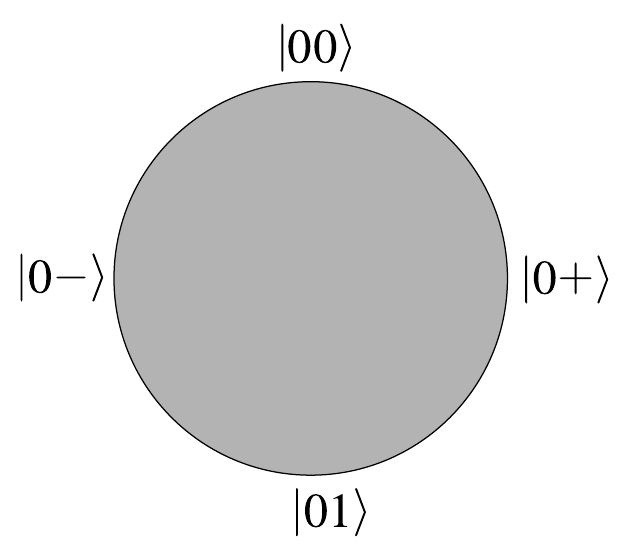}
\caption{\label{fig_sep}\textbf{Bloch sphere of the example for
Class~\ref{class_sep}}: all the states on and inside this Bloch sphere
are separable.}
\end{figure}

\begin{figure}
\includegraphics{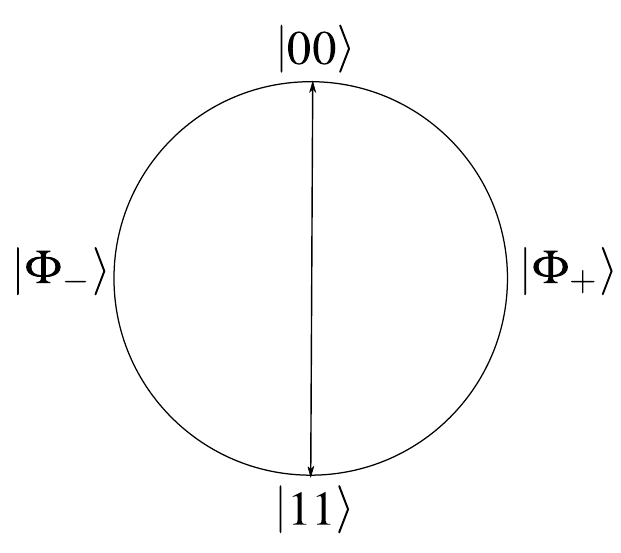}
\caption{\label{fig_line_ortho}\textbf{Bloch sphere of the example for
Class~\ref{class_line_ortho}}: all the states along the line connecting
$\ket{00}$ and $\ket{11}$ are separable;
all the other states on and inside this Bloch sphere are entangled.
Any two orthogonal product states
can replace $\ket{00}$ and $\ket{11}$.}
\end{figure}

\begin{figure}
\includegraphics{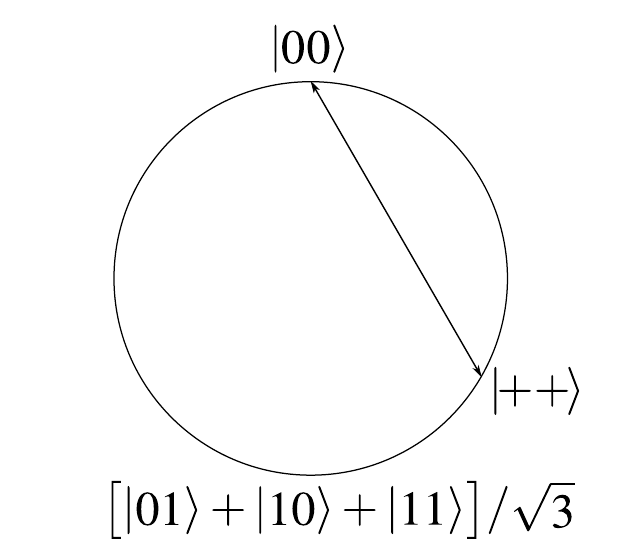}
\caption{\label{fig_line_nonortho}\textbf{Bloch sphere of the example
for Class~\ref{class_line_nonortho}}: all the states along the line
connecting $\ket{00}$ and $\ket{++}$ are separable;
all the other states on and inside this Bloch sphere are entangled.
Any two non-orthogonal linearly independent product states
can replace $\ket{00}$ and $\ket{++}$.}
\end{figure}

We now analyze the example for Class~\ref{class_point}
(see Fig.~\ref{fig_point}), a class that we found,
yet was also found independently by~\cite{Regula2016}.
\begin{proposition}\label{prop_class4_example}
Let $\ket{\psi_0} = \ket{00}$ and
$\ket{\psi_1} = \alpha\ket{01} + \beta\ket{10}$ with
$\alpha\beta \neq 0$, $|\alpha|^2 + |\beta|^2 = 1$.
The state
$\rho = a_{00}\ket{\psi_0}\bra{\psi_0} + a_{01}\ket{\psi_0}\bra{\psi_1}
+ a_{10}\ket{\psi_1}\bra{\psi_0} + a_{11}\ket{\psi_1}\bra{\psi_1}$
is separable if and only if $a_{01} = a_{10} = a_{11} = 0$.
\end{proposition}
\begin{proof}
$\bra{11}\rho\ket{11} = 0$ and
$\bra{10}\rho\ket{01} = a_{11}\braket{10}{\psi_1}\braket{\psi_1}{01}
= a_{11}\beta\alpha^*$;
thus, by Lemma~\ref{PeresWeaker}
(the ``weaker entanglement criterion''),
$\rho$ is entangled if $a_{11} \neq 0$. 
If $a_{11} = 0$, $\bra{\psi_1}\rho\ket{\psi_1} = 0$ and
$\bra{\psi_1}\rho\ket{\psi_0} = a_{10}$;
therefore, by Lemma~\ref{zelemma}, $a_{10} = 0$,
which implies that $a_{01} = a_{10}^* = 0$.
\end{proof}

\begin{figure}
\includegraphics{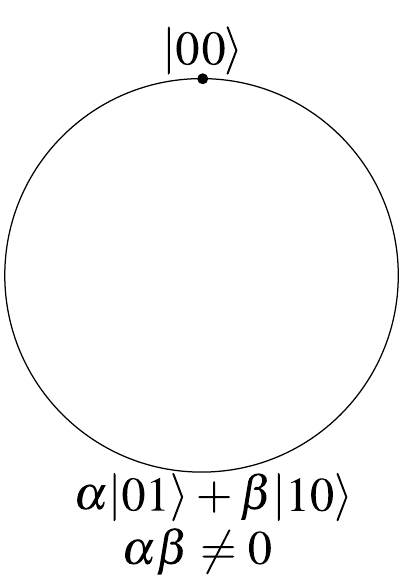}
\caption{\label{fig_point}\textbf{Bloch sphere of the example for
Class~\ref{class_point}}: only the state $\ket{00}$ is separable;
all the other states on and inside this Bloch sphere are entangled.}
\end{figure}

Finally, for the example of Class~\ref{class_entang}
(see Fig.~\ref{fig_entang}):
\begin{proposition}\label{prop_class5_example}
Let $\ket{\psi_0} = (\ket{00}+\ket{11})/\sqrt{2}$ and
$\ket{\psi_1} = (\ket{02}+\ket{10})/\sqrt{2}$.
The state
$\rho = a_{00}\ket{\psi_0}\bra{\psi_0} + a_{01}\ket{\psi_0}\bra{\psi_1}
+ a_{10}\ket{\psi_1}\bra{\psi_0} + a_{11}\ket{\psi_1}\bra{\psi_1}$
is \underline{always} entangled.
\end{proposition}
\begin{proof}
By using Corollary~\ref{corollary_separable}
(the ``separable states condition''),
it is sufficient to prove that all the pure states
$\ket{\psi} = \alpha \ket{\psi_0} + \beta \ket{\psi_1}$
are entangled.

Let us look at the state
\begin{eqnarray*}
\ket{\psi} &=& \alpha \ket{\psi_0} + \beta \ket{\psi_1} \\
&=& \frac{\alpha}{\sqrt{2}} \ket{00}
+ \frac{\alpha}{\sqrt{2}} \ket{11}
+ \frac{\beta}{\sqrt{2}} \ket{02}
+ \frac{\beta}{\sqrt{2}} \ket{10} \\
&\triangleq& \sum_{i,j} \epsilon_{ij} \ket{i}\ket{j}.
\end{eqnarray*}
For this state to be separable, there must exist
$a_0, a_1$ and $b_0, b_1, b_2$ such that
$\epsilon_{ij} = a_i b_j$ for all $i,j$;
hence, the equations $\epsilon_{01} = a_0 b_1 = 0$ and
$\epsilon_{12} = a_1 b_2 = 0$ must hold.
By a simple calculation it follows that necessarily
$\alpha = \beta = 0$, which is impossible.
We conclude that there are no separable pure states on the Bloch sphere,
and thus there are no separable mixed states.
\end{proof}

\begin{figure}
\includegraphics{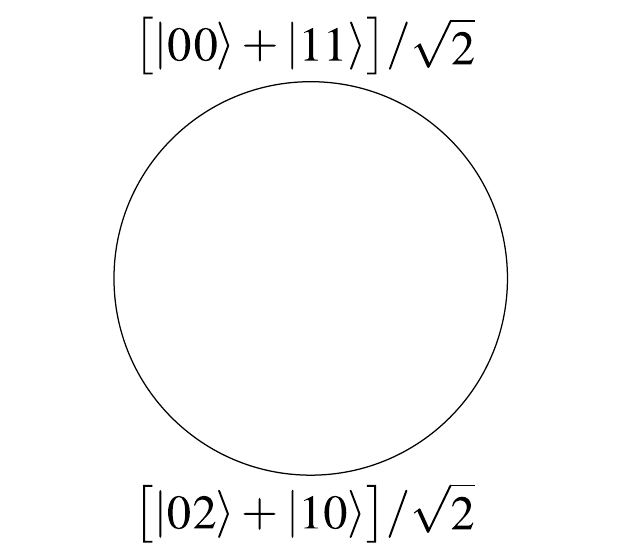}
\caption{\label{fig_entang}\textbf{Bloch sphere of the example for
Class~\ref{class_entang}}: all the states on and inside this Bloch
sphere are entangled.}
\end{figure}

Our classification suggests natural ways to measure entanglement
inside the Bloch sphere: for example, entanglement may be measured
by the Euclidean distance to the closest separable state (e.g.,
given the Bloch sphere $\text{Span}\{\ket{00}, \ket{11}\}$,
the closest separable state to the pure state
$\alpha \ket{00} + \beta \ket{11}$
is the state
$|\alpha|^2 \ket{00}\bra{00} + |\beta|^2 \ket{11}\bra{11}$).
We note that this entanglement measure,
unlike the measures analyzed by~\cite{Lohmayer2006,Osterloh2008},
vanishes only for separable states.
Analyzing the properties of such measures
is beyond the scope of this paper.

\section{\label{proof_central} A proof that there are
exactly five classes of ``Bloch-sphere entanglement''}
Our main goal is to provide a full analysis of the
general bipartite case. We prove that the classes we found
are the \underline{only} classes that exist in the bipartite case,
for all the rank-2 bipartite states
(namely, for all the corresponding $2$-dimensional Hilbert spaces):

\begin{theorem}\label{central_theorem}
Let $\mathcal{H}$ be a $2$-dimensional subspace of
$\mathcal{H}_A \otimes \mathcal{H}_B$,
where $\mathcal{H}_A$ and $\mathcal{H}_B$ are two Hilbert spaces.
Then $\mathcal{H}$ belongs to one of the following classes:
\begin{description}
	\item[Class~\ref{class_sep}] The Bloch ball of $\mathcal{H}$ is
	\underline{completely separable}.
	\item[Classes~\ref{class_line_ortho}+\ref{class_line_nonortho}]
	The Bloch ball of $\mathcal{H}$ has
	\underline{one line} of separable states,
	and all the other states are entangled.
	\item[Class~\ref{class_point}] The Bloch ball of $\mathcal{H}$ has
	\underline{one separable point} (pure state),
	and all the other states are entangled.
	\item[Class~\ref{class_entang}] The Bloch ball of $\mathcal{H}$ is
	\underline{completely entangled}.
\end{description}
\end{theorem}
(We note that Class~\ref{class_line_ortho}
and Class~\ref{class_line_nonortho} are discussed together,
because in both of them the Bloch ball has
just one line of separable states.)

\begin{proof}
First, assume that there is no separable
\emph{mixed} state inside the Bloch ball.
This means that there is at most one
pure separable state on the Bloch sphere
(because if two pure states are separable,
then the line connecting them inside the Bloch ball is separable, too).
This matches Classes~\ref{class_point} and \ref{class_entang}.

Now assume that there is a separable
\emph{mixed} state $\rho$ inside the Bloch ball.
According to Corollary~\ref{corollary_separable}
(the ``separable states condition''),
this means that there are at least two
different pure separable states on the Bloch sphere.
We denote them by $\ket{\psi} = \ket{\psi_A} \otimes \ket{\psi_B}$ and
$\ket{\phi} = \ket{\phi_A} \otimes \ket{\phi_B}$.

We note that $\ket{\psi} \ncong \ket{\phi}$
(defining the symbol $\cong$ to be ``equality as normalized states,
possibly with different global phases''; thus, the symbol $\ncong$
means that the two normalized states are really different,
as opposed to states that are equal up to a global phase),
which means that $\ket{\psi}$ and $\ket{\phi}$ are linearly independent.
Therefore, the Bloch sphere represents the $2$-dimensional subspace
$\text{Span}\{\ket{\psi}, \ket{\phi}\}$, which means that
all the mixed states inside the Bloch ball are of the form:
\begin{equation}
\rho = a_{00} \ket{\psi} \bra{\psi} + a_{01} \ket{\psi} \bra{\phi}
+ a_{10} \ket{\phi} \bra{\psi} + a_{11} \ket{\phi} \bra{\phi}
\label{class2_state}
\end{equation}

If $\ket{\psi_A} \cong \ket{\phi_A}$ or
$\ket{\psi_B} \cong \ket{\phi_B}$,
then obviously all the states on and inside the Bloch sphere
are separable, which matches Class~\ref{class_sep}.

If $\ket{\psi_A} \ncong \ket{\phi_A}$
and $\ket{\psi_B} \ncong \ket{\phi_B}$,
then we prove that only the line connecting
$\ket{\psi}$ and $\ket{\phi}$ inside the Bloch ball is separable,
and that all the other pure and mixed states in the Bloch ball
are entangled. This will match
Classes~\ref{class_line_ortho}+\ref{class_line_nonortho},
and will conclude our proof.

We look at all the mixed states of the form
\eqref{class2_state}.
If $a_{01} = a_{10} = 0$, then we obviously get a separable state:
\begin{equation*}
\rho = a_{00} \ket{\psi_A}\bra{\psi_A}
\otimes \ket{\psi_B}\bra{\psi_B}
+ a_{11} \ket{\phi_A}\bra{\phi_A} \otimes \ket{\phi_B}\bra{\phi_B}
\end{equation*}

If $a_{10} \ne 0$, then: let $\ket{\overline{\phi_A}} \in \mathcal{H}_A$
satisfy $\braket{\phi_A}{\overline{\phi_A}} = 0$ and
$\braket{\psi_A}{\overline{\phi_A}} \ne 0$
($\ket{\overline{\phi_A}}$ always exists,
because $\ket{\psi_A} \ncong \ket{\phi_A}$).
Similarly, let $\ket{\overline{\psi_A}} \in \mathcal{H}_A$ and
$\ket{\overline{\phi_B}}, \ket{\overline{\psi_B}} \in \mathcal{H}_B$
satisfy similar properties (because $\ket{\psi_B} \ncong \ket{\phi_B}$).
Then
\begin{equation*}
\bra{\overline{\psi_A} \; \overline{\phi_B}} \rho
\ket{\overline{\psi_A} \; \overline{\phi_B}} = 0,
\end{equation*}
and
\begin{eqnarray*}
\bra{\overline{\psi_A} \; \overline{\psi_B}} \rho
\ket{\overline{\phi_A} \; \overline{\phi_B}}
&=& a_{10} \braket{\overline{\psi_A}
\; \overline{\psi_B}}{\phi_A \phi_B}
\braket{\psi_A \psi_B}{\overline{\phi_A}
\; \overline{\phi_B}} \\
&=& a_{10} \braket{\overline{\psi_A}}{\phi_A}
\braket{\overline{\psi_B}}{\phi_B}
\braket{\psi_A}{\overline{\phi_A}}
\braket{\psi_B}{\overline{\phi_B}} \\
&\ne& 0.
\end{eqnarray*}

Therefore, by Lemma~\ref{PeresWeaker}
(the ``weaker entanglement criterion''),
if $a_{10} \ne 0$ (or $a_{01} \ne 0$
-- this is equivalent, because $a_{01} = a_{10}^*$),
then $\rho$ is entangled.

We conclude that only the line between $\ket{\psi}$ and $\ket{\phi}$
(i.e., the line of states satisfying $a_{01} = a_{10} = 0$)
is separable, and that the other states (i.e., the states satisfying
$a_{10} \ne 0$ or $a_{01} \ne 0$) are entangled, which matches
Classes~\ref{class_line_ortho}+\ref{class_line_nonortho}.
This concludes our proof.
\end{proof}

\section{\label{two_qubit_proof} A proof that Class~\ref{class_entang}
does not exist in the two-qubit case}
We have seen that for almost all the bipartite Hilbert spaces,
five classes appear. We now show that for the Hilbert space
$\mathcal{H}_2 \otimes \mathcal{H}_2$, only four classes exist
(Classes~\ref{class_sep}-\ref{class_point}):

\begin{proposition}\label{prop_two_qubit}
No $2$-dimensional subspace of $\mathcal{H}_2\otimes\mathcal{H}_2$
is completely entangled.
\end{proposition}
\begin{proof}
This proof follows the methods of~\cite{Osterloh2008}.
We remember that for a two-qubit state
$\ket{\psi} = \sum_{i,j} a_{ij} \ket{i}\ket{j}$, the concurrence $C$ is
defined as follows \cite{HW1997,Wootters2001}:
\begin{equation}
C(\psi) = 2 |a_{00} a_{11} - a_{01} a_{10}|
\end{equation}
In particular, $C(\psi) = 0$ if and only if $\ket{\psi}$ is separable.
(This is not necessarily true for other entanglement measures.)

Let $\mathcal{H} \triangleq \text{Span}\{\ket{\psi_0}, \ket{\psi_1}\}$
be a $2$-dimensional subspace of $\mathcal{H}_2 \otimes \mathcal{H}_2$.
We may assume that $C(\psi_1) \ne 0$ (otherwise, $\ket{\psi_1}$ is
separable, hence $\mathcal{H}$ cannot be completely entangled).
Therefore, the set of separable (non-normalized) pure states
in $\mathcal{H}$ is the set of states $\ket{\psi_0} + z \ket{\psi_1}$
satisfying the following equation:
\begin{equation*}
C(\ket{\psi_0} + z \ket{\psi_1}) = 0
\end{equation*}

This is a quadratic equation in the complex variable $z$
(because we may ignore the absolute value).
The absolute value of the coefficient of $z^2$ is $C(\psi_1) \ne 0$.
Therefore, there are two solutions $\xi_1, \xi_2$
(possibly equal) to this equation,
and thus the non-normalized state $\ket{\psi_0} + \xi_1 \ket{\psi_1}$
(whose normalization is in $\mathcal{H}$) must be separable.
Therefore, there is a separable state in $\mathcal{H}$,
and $\mathcal{H}$ cannot be completely entangled.
\end{proof}

\section{\label{multipartite_ent} Examples and analysis of
multipartite entanglement}
For multipartite states, there are several different definitions of
separability and entanglement: an $m$-partite mixed state is
``fully separable'' if it is a mixture of pure states that are products
of $m$ pure states; and it is ``separable with respect to a bipartite
partition $\mathcal{P}$'' (with $\mathcal{P}$ partitioning
the $m$ subsystems into two disjoint sets)
if the bipartite state corresponding to the partition
$\mathcal{P}$ is separable~\cite{Horodecki2009}.
For example, the state $\ket{0}_A \ket{\Phi_+}_{BC} \in
\mathcal{H}_A \otimes \mathcal{H}_B \otimes \mathcal{H}_C$
is separable with respect to the partition $\{\{1\}, \{2, 3\}\}$,
but is entangled with respect to both partitions $\{\{1, 2\}, \{3\}\}$
and $\{\{1, 3\}, \{2\}\}$.
Note that even if a state is separable with respect to all
the bipartite partitions, it may still be entangled
(i.e., not fully separable) \cite{BDMSST1999}.

To illustrate the many existing possibilities for Bloch spheres
in the multipartite case, we look at two examples:
\begin{enumerate}
\item $\text{Span} \{\ket{000}, \ket{111}\}$: the line connecting
between the north pole ($\ket{000}$) and the south pole ($\ket{111}$)
is fully separable; all the other points are entangled
with respect to \emph{any} bipartite partition.
\item $\text{Span} \{\ket{000}, \ket{011}\}$: the line connecting
between the north pole ($\ket{000}$) and the south pole ($\ket{011}$)
is fully separable; all the other points are separable with respect to
the bipartite partition $\{\{1\}, \{2,3\}\}$, but are entangled
with respect to the partitions $\{\{1,2\}, \{3\}\}$
and $\{\{1,3\}, \{2\}\}$.
\end{enumerate}

The proofs of separability above are direct from the definitions;
and the proofs of entanglement are implied by our analysis
in the proof of Theorem~\ref{central_theorem}.

Moreover, our Theorem~\ref{central_theorem} is true also for the
set of \emph{fully separable} states in the multipartite case:

\begin{theorem}\label{central_theorem_multipartite}
Let $\mathcal{H}$ be a $2$-dimensional subspace of
$\mathcal{H}_{A_1} \otimes \cdots \otimes \mathcal{H}_{A_m}$,
where $\mathcal{H}_{A_1}, \dots, \mathcal{H}_{A_m}$ are Hilbert spaces.
Then $\mathcal{H}$ belongs to one of the following classes:
\begin{description}
	\item[Class~\ref{class_sep}] \underline{All} the states inside
	the Bloch ball of $\mathcal{H}$ are \underline{fully separable}.
	\item[Classes~\ref{class_line_ortho}+\ref{class_line_nonortho}]
	The Bloch ball of $\mathcal{H}$ has
	\underline{one line} of fully separable states,
	and all the other states are not fully separable.
	\item[Class~\ref{class_point}] The Bloch ball of $\mathcal{H}$ has
	\underline{one fully separable point} (pure state),
	and all the other states are not fully separable.
	\item[Class~\ref{class_entang}] \underline{All} the states inside
	the Bloch ball of $\mathcal{H}$ are \underline{not fully separable}.
\end{description}
\end{theorem}

\begin{proof}
First, assume that there is no fully separable
\emph{mixed} state inside the Bloch ball.
This means that there is at most one
pure fully-separable state on the Bloch sphere
(because if two pure states are fully separable,
then the line connecting them inside
the Bloch ball is fully separable, too).
This matches Classes~\ref{class_point} and~\ref{class_entang}.

Now assume that there is a fully separable
\emph{mixed} state $\rho$ inside the Bloch ball.
According to Corollary~\ref{corollary_separable}
(the ``separable states condition''), this means
that there are at least two different fully separable
pure states on the Bloch sphere. We denote them by
$\ket{\psi} = \ket{\psi_{A_1}} \otimes \cdots \otimes \ket{\psi_{A_m}}$
and
$\ket{\phi} = \ket{\phi_{A_1}} \otimes \cdots \otimes \ket{\phi_{A_m}}$.

We note that $\ket{\psi} \ncong \ket{\phi}$
(defining the symbol $\cong$ as we did in the proof of
Theorem~\ref{central_theorem} above; thus, the symbol
$\ncong$ means that the two normalized states are really different,
as opposed to states that are equal up to a global phase),
which means that $\ket{\psi}$ and $\ket{\phi}$ are linearly independent.
Therefore, the Bloch sphere represents the $2$-dimensional subspace
$\text{Span}\{\ket{\psi}, \ket{\phi}\}$, which means that
all the mixed states inside the Bloch ball are of the form:
\begin{equation}
	\rho = a_{00} \ket{\psi} \bra{\psi} + a_{01} \ket{\psi} \bra{\phi}
	+ a_{10} \ket{\phi} \bra{\psi} + a_{11} \ket{\phi} \bra{\phi}
\end{equation}

If $\ket{\psi_{A_i}} \cong \ket{\phi_{A_i}}$ for all $i$
\underline{except} one value of $i$, then obviously all the states
on and inside the Bloch sphere are fully separable,
which matches Class~\ref{class_sep}.

If $\ket{\psi_{A_{i_1}}} \ncong \ket{\phi_{A_{i_1}}}$ and
$\ket{\psi_{A_{i_2}}} \ncong \ket{\phi_{A_{i_2}}}$ for $i_1 < i_2$,
then we prove that for the bipartite partition $\{I_1, I_2\}$
with $I_1 = \{1, \dots, i_1\}$ and $I_2 = \{i_1 + 1, \dots, m\}$
(satisfying $I_1 \cup I_2 = \{1, \dots, m\}$,
$I_1 \cap I_2 = \emptyset$, $i_1 \in I_1$, and $i_2 \in I_2$),
it holds that only the line connecting $\ket{\psi}$ and $\ket{\phi}$
inside the Bloch ball is fully separable,
and that all the other pure and mixed states in the Bloch ball
are entangled with respect to the partition $\{I_1, I_2\}$.
This will match Classes~\ref{class_line_ortho}+\ref{class_line_nonortho},
and will conclude our proof.

To prove that the line is fully separable, we notice that
any convex combination of fully separable states is fully separable,
and therefore the line connecting $\ket{\psi}$ and
$\ket{\phi}$ inside the Bloch ball is fully separable.

To prove that all the other states are entangled with respect to
the partition $\{I_1, I_2\}$, we denote $\ket{\psi^{I_1}}
= \ket{\psi_{A_1}} \otimes \cdots \otimes \ket{\psi_{A_{i_1}}}$
and $\ket{\psi^{I_2}} = \ket{\psi_{A_{i_1 + 1}}} \otimes
\cdots \otimes \ket{\psi_{A_m}}$; and similarly, we define
$\ket{\phi^{I_1}}$ and $\ket{\phi^{I_2}}$.
Then, because $i_1 < i_2$, and because
$\ket{\psi_{A_{i_1}}} \ncong \ket{\phi_{A_{i_1}}}$ and
$\ket{\psi_{A_{i_2}}} \ncong \ket{\phi_{A_{i_2}}}$,
it must hold that $\ket{\psi^{I_1}} \ncong \ket{\phi^{I_1}}$ and
$\ket{\psi^{I_2}} \ncong \ket{\phi^{I_2}}$. It also holds that
$\ket{\psi} = \ket{\psi^{I_1}} \otimes \ket{\psi^{I_2}}$ and
$\ket{\phi} = \ket{\phi^{I_1}} \otimes \ket{\phi^{I_2}}$; therefore,
according to the proof of the original Theorem~\ref{central_theorem},
it holds that all the states outside of the line
connecting $\ket{\psi}$ and $\ket{\phi}$ in the Bloch ball
(i.e., all the states satisfying $a_{01} \ne 0$ or $a_{10} \ne 0$)
are entangled with respect to the partition $\{I_1, I_2\}$.
Together with the proof that all the states on that line (i.e.,
all the states satisfying $a_{01} = a_{10} = 0$) are fully separable,
this matches Classes~\ref{class_line_ortho}+\ref{class_line_nonortho},
and concludes our proof.
\end{proof}

Extensions of Theorem~\ref{central_theorem} to other cases of
multipartite entanglement are beyond the scope of this paper.

\section{\label{prev_works} Previous works}
The existence of completely entangled subspaces has been discussed
in many papers before. In particular, this notion was used
in~\cite{BDMSST1999} to prove the existence of a huge class
of bound entangled states.

Analysis of entangled states in a Hilbert subspace,
using \emph{specific} entanglement measures
(e.g., the concurrence and the 3-tangle) and Bloch spheres,
was done by~\cite{Lohmayer2006} and~\cite{Osterloh2008}.
However, the entanglement measures they choose usually vanish
not only for all the separable states,
but also for some of the entangled states~\cite{Osterloh2008}.
Much more recently,~\cite{Regula2016} and~\cite{Regula2016_ii}
investigated interesting classes in the \emph{same} research direction.
In contrast, our paper analyzes the separability and the entanglement
in the Bloch sphere for any rank-2 bipartite state;
and, instead of using a specific entanglement measure
that \emph{cannot show the entanglement}
of some of the entangled states,
we fully characterize the set of separable states
on and inside the state's Bloch sphere.

\section{\label{conclusion} Conclusion}
We have found a complete classification of the possible
sets of separable states in all the $2$-dimensional subspaces
of bipartite Hilbert spaces.
Our result is general and is not limited to specific entanglement
measures or to specific bipartite spaces, but applies to all
the bipartite Hilbert spaces, and extends to the sets of
fully separable states in multipartite spaces.
Moreover, the result makes it possible to define natural measures
that vanish exactly on the separable states.

It may be possible to extend our results into higher-rank mixed states:
for example, it is possible to look at ``portions''
of the higher-rank states (e.g., a non-degenerate rank-3 state
defines three Bloch spheres, each corresponding to
two out of the three eigenstates);
and it is possible to analyze higher-rank states that are
$\epsilon$-close ($\epsilon \ll 1$) to rank-2 states.

Our analysis identifies the set of ``Bloch-sphere neighbor states''
of any rank-2 state (namely, the set of states in its Bloch sphere).
Such Bloch-sphere neighbor states may be useful for various protocols:
for example, entanglement purification or error correction protocols
may first turn the state into a Bloch-sphere neighbor state of
desired properties (e.g., more entangled), and then operate on that
Bloch-sphere neighbor state. Those possibilities may be explored
by future research.

\begin{acknowledgments}
We thank Ajit Iqbal Singh for telling us that Class~\ref{class_point}
we present was previously shown (independently of our result)
in~\cite{Regula2016}.
The work of TM and RL was partly supported
by the Israeli MOD Research and Technology Unit.
\end{acknowledgments}

\bibliography{Geometry}

\end{document}